\documentclass[a4paper,12pt]{article}
\usepackage{amsmath, amsthm, amstext}
\usepackage{amssymb, array, amsfonts}
\usepackage[cp1251,cp866]{inputenc}
\usepackage[english,russian]{babel}
\usepackage{graphicx}
\inputencoding{cp1251}

\numberwithin{equation}{section}

\def \de{\delta}
\def \er{\varepsilon}

\def \la{\lambda}

\def \te{\theta}
\def \ph{\varphi}

\def \D{\Delta}

\def \C{\mathbb{C}}

\def \R{\mathbb{R}}

\def\n{\nabla}
\def\dd{\partial}
\def\div{\operatorname{div}}
\def\rot{\operatorname{rot}}
\def\1{1\!\!\!\!1}

\def\dom{\operatorname{Dom}}

\newcommand{\<}{\langle}
\renewcommand{\>}{\rangle}

\theoremstyle{plain}
\newtheorem{theorem}{\bf Theorem}[section]
\newtheorem{lemma}[theorem]{\bf Lemma}

\newtheorem{cor}[theorem]{\bf Corollary}

\theoremstyle{definition}

\theoremstyle{remark}
\newtheorem{rem}[theorem]{\bf Remark}

\renewcommand{\le}{\leqslant}
\renewcommand{\ge}{\geqslant}

\renewcommand{\proofname}{Proof}

\renewcommand\refname{References}

\title{The Maxwell operator with periodic coefficients in a cylinder}

\author{N.~Filonov, A.~Prokhorov
\thanks{The first author was supported by the grant RNF 17-11-01126.
The second author was supported by the grant RFBR 14-01-00760.}
}
\date{}
\begin{document}
	\renewcommand\refname{References}
	\renewcommand{\abstractname}{Abstract}
	\renewcommand{\proofname}{Proof}
\maketitle

\begin{abstract}
In the paper we consider the Maxwell operator in a three-dimensional cylinder with coefficients periodic along the axis of a cylinder. 
It is proved that for cylinders with circular and rectangular cross-section the spectrum of the Maxwell operator is absolutely continuous.
\end{abstract}

\section*{Introduction}
Let $U\subset \R^2$ be a bounded domain and
$\Pi = U \times \R$ be a three-dimensional cylinder.
Let $\er$ and $\mu$ be two scalar functions describing the dielectric and magnetic permeabilities of the medium which fills the cylinder $\Pi$.
We will always assume that these functions are bounded and separated from zero:
\begin{equation}
\label{01}
0 < \er_0 \le \er(x) \le \er_1 , \quad 
0 < \mu_0 \le \mu(x) \le \mu_1 , \quad x \in \overline\Pi .
\end{equation}
The Maxwell operator (see, for example, \cite{BS89}) 
acts by the formula
\begin{equation}
\label{02}
{\cal M} 
\left( \begin{array}{cc} u \\ v \end{array} \right) =
\left( \begin{array}{cc}
i \er^{-1} \rot v \\ -i \mu^{-1} \rot u 
\end{array} \right)
\end{equation}
on the pairs of vector functions $\{u,v\}$ defined on $\Pi$, satisfying the divergence free condition 
\begin{equation}
\label{03}
\div (\er u) = \div (\mu v) = 0
\end{equation}
and the perfect conductivity condition on the boundary of the cylinder 
\begin{equation}
\label{04}
\left.u_\tau \right|_{\dd\Pi} = 0, \quad 
\left.v_\nu \right|_{\dd\Pi} = 0 ;
\end{equation}
the subscripts $\tau$ and $\nu$ denote the tangent and normal components of vector respectively.
Functions $u$ and $v$ have the meaning of electric and magnetic field components in the cylinder.
The exact definition of the operator $\cal M$ is given below in section 1.
The Maxwell operator is self-adjoint in a suitable Hilbert space.
We are interested in the structure of its spectrum when the coefficients  $\er$ and $\mu$ are periodic along the axis of the cylinder.
\begin{equation}
\label{05}
\er (x + e_3) = \er (x), \quad \mu (x + e_3) = \mu (x), \quad x \in \overline\Pi .
\end{equation}
It is well known that for operators with periodic coefficients the spectrum has a band structure (see, for example, \cite{Ku}), there is no singular continuous component (see also \cite{FilSob}), and the eigenvalues can be only of infinite multiplicity (degenerate bands). Our goal is to establish the absolute continuity of the spectrum of the Maxwell operator or, equivalently, the absence of its eigenvalues.

In the paper of A.~Morame absolute continuity of the Maxwell operator was established in the whole space $\R^3$ in the case of scalar coefficients $\er, \mu \in C^\infty$ periodic with respect to three-dimensional lattice. T.~A.~Suslina simplified the proof of Morame, which made possible to relax the smoothness conditions to
\begin{equation}
\label{06}
\er, \mu \in W^2_{3/2, loc} .
\end{equation}
In addition, Suslina proved the absolute continuity of the Maxwell operator in 
a layer $[0,a] \times \R^2$ under the same condition \eqref{06}, see \cite{Su} (formally speaking, it's required in \cite{Su} that $\er, \mu \in W^2_{3/2, loc} \cap W^1_{p, loc}$,
$p>3$, but the proof works in the case $p=3$, and then \eqref{06} is sufficient by embedding theorem $W^2_{3/2} \subset W^1_3$). In the paper \cite{FKl} the absolute continuity of the Maxwell operator in $\R^3$ was established in the case when the coefficients are periodic along certain
directions
and tending to constants super-exponentially fast in the remaining directions
(the so-called soft waveguide).

We show that in the case of cylinders with circular and rectangular cross-section the spectrum of the Maxwell operator is absolutely continuous. 
We also assume that the coefficients $ \er $ and $ \mu $ are scalar functions (which corresponds to an isotropic medium), rather than matrix-valued functions, and sufficiently smooth. The question of what happens if at least one of these conditions is removed is open already in the case of $ \R^3 $. It is known that if we abandon both conditions, that is, consider an operator with nonsmooth matrix coefficients, then the spectrum {\it can} have eigenvalues of infinite multiplicity (see \cite{D}).

In all three papers \cite{Mo}, \cite{Su} and \cite{FKl} it is proved that an eigenvalue of the Maxwell operator \eqref{02}--\eqref{04} is possible only if there is an eigenvalue of some Schr\"odinger operator $-\D + V$. The function $V$ is periodic, $(6\times 6)$--matrix-valued and, generally speaking, non-self-adjoint. In the case of a layer the Robin boundary condition arises with periodic coefficients that are a combination of
$ \er $, $ \mu $ and their derivatives. We follow the same pattern. In this paper we show that in the case of a cylinder with quite arbitrary cross-section, the question of the absolute continuity of
the Maxwell operator also reduces to the question of the absence of eigenvalues for some Schr\"odinger operator in a cylinder with the Robin boundary condition. Moreover, the terms related to the curvature of the boundary are added to the coefficients in the boundary condition. The result on the absolute continuity of the Maxwell operator is obtained only for cylinders with rectangular or circular cross-section, since currently the corresponding theorems for the Schr\"odinger operator are obtained only for such cylinders, see \cite {K1} and \cite {K2}.

\section{Formulation of results}
Let $ \er $, $ \mu $ be real scalar functions in the cylinder $ \Pi = U \times \R $ satisfying the conditions \eqref {01} and \eqref {05}.
We introduce the Hilbert space
\begin{equation}
\label{10}
J = \left\{(u;v) \in L_2 (\Pi, \C^3, \er dx) \oplus 
L_2 (\Pi, \C^3, \mu dx) : \div(\er u) = \div (\mu v) = 0, 
\left.v_\nu\right|_{\dd\Pi} = 0\right\} .
\end{equation}
Conditions on $ u $ and $ v $ are understood in the sense of integral identities
$$
\div(\er u) = 0 \quad \Leftrightarrow \quad 
\int_\Pi \< \er u, \n \eta\>\, dx = 0 \quad \forall \eta \in H^1_0 (\Pi, \C),
$$
$$
\div(\mu v) = 0, \left.v_\nu\right|_{\dd\Pi} = 0 
\quad \Leftrightarrow \quad 
\int_\Pi \< \mu v, \n \te\>\, dx = 0 \quad \forall \te \in H^1 (\Pi, \C),
$$
where $ H^1 $, $ H^1_0 $ are the Sobolev space;
$ H^1_0 (\Pi) $ is the closure of the set $ C_0^\infty (\Pi) $ in $ H^1 (\Pi) $.
In the space $ J $ we consider the Maxwell operator
\begin{equation}
\label{11}
{\cal M} =
\left( \begin{array}{cc}
0 & i \er^{-1} \rot \\
-i \mu^{-1} \rot & 0
\end{array} \right) 
\end{equation}
on the domain
\begin{equation}
\label{12}
\dom {\cal M} = 
\{(u;v) \in J : \rot u, \rot v \in L_2 (\Pi, \C^3), 
\left.u_\tau\right|_{\dd\Pi} = 0\} .
\end{equation}
The boundary condition is again defined in terms of the integral identity
$$
\left.u_\tau\right|_{\dd\Pi} = 0
\quad \Leftrightarrow \quad 
\int_\Pi \<\rot u, z\> dx = \int_\Pi \<u, \rot z\> dx \quad 
\forall \ z \in L_2(\Pi, \C^3) : \rot z \in L_2 (\Pi, \C^3) .
$$
It is well known (see \cite {BS89}) that the so defined operator $ {\cal M} $ is self-adjoint.

We introduce the notation $ \widetilde {L_p} (\Pi) $ for the space of periodic
functions $ f $, 
$$
f(x+e_3) = f(x),
$$ 
such that $ f \in L_p (\Pi \cap B_R) $ for any $ R $, where $ B_R $ is a ball of radius $ R $.
The symbols $ \widetilde {W_p ^ l} (\Pi) $ and $ \widetilde {L_p} (\dd  \Pi) $ have a similar meaning.

In the previous work of the authors \cite [Theorem 1.4] {FP} is established

\begin{theorem}
\label{t11}
Let $ \Pi = U \times \R $, where the cross-section $ U $ is a convex bounded domain on the plane. Let the coefficients $ \er $, $ \mu $ satisfy \eqref{01} and $ \er, \mu \in \widetilde W^1_3 (\Pi) $. Then the set $ \dom {\cal M} $ admits an equivalent description
$$
\dom {\cal M} = \left\{ (u;v) \in H^1 (\Pi, \C) :
\div(\er u) = \div (\mu v) = 0,
\left.u_\tau\right|_{\dd\Pi} = 0,
\left.v_\nu\right|_{\dd\Pi} = 0\right\} ;
$$
here the boundary conditions can be understood in the sense of traces.
\end{theorem}

Thus, the "weak"\ Maxwell operator coincides with the "strong"\ Maxwell operator.

\begin{rem}
\label{r12}
The claim of the Theorem \ref {t11} remains valid if the convexity condition of the domain $ U $ is replaced by the smoothness condition $ \dd U \in W^2_p $, $ p> 2 $.
\end{rem}

We introduce the space
$$
\hat H^1 (\Pi) = 
\{\Phi \in H^1 (\Pi, \C^6) : \left.N \Phi\right|_{\dd\Pi} = 0\} ,
$$
where
$$
N = 
\left( \begin{array}{cccccc}
0 & -\nu_3 & \nu_2 & 0 & 0 & 0 \\
\nu_3 & 0 & - \nu_1 & 0 & 0 & 0 \\
-\nu_2 & \nu_1 & 0 & 0 & 0 & 0 \\
0 & 0 & 0 & \nu_1 & \nu_2 & \nu_3
\end{array} \right) ,
$$
$\nu(x)$ is the normal vector at the point $x\in\dd\Pi$.

We introduce an additional restriction on the cylinder $ \Pi $. Let $ q \ge 3/2 $, $ r \ge 2 $.

\vspace{4pt}
\noindent {\bf Condition} $A(q,r)$. 
{\it The domain $\Pi$ is such that conditions}
$$
\Phi \in \hat H^1 (\Pi), \qquad 
V \in \widetilde{L_q} (\Pi, \operatorname{Mat} (6\times 6, \C)), \qquad
\Sigma \in \widetilde{L_r} (\dd\Pi, \operatorname{Mat} (6\times 6, \C))
$$
{\it and integral identity}
\begin{equation}
\label{13}
\int_\Pi \left(\< \dd_j \Phi, \dd_j \Psi \>_{\C^6}
+ \< V \Phi, \Psi \>_{\C^6}\right) dx +
\int_{\dd\Pi} \< \Sigma \Phi, \Psi \>_{\C^6} dS(x) = 0,
\quad \forall \Psi \in \hat H^1 (\Pi),
\end{equation}
{\it imply} $\Phi \equiv 0$.

Hereinafter, summation over repeated indices from 1 to 3 is implied.

\begin{rem}
The identity \eqref {13} means that the function $ \Phi $ lies in the kernel of the matrix Schr\"odinger operator $ - \D + V $ in the cylinder $ \Pi $ with boundary conditions
$$
\left.\Phi_\tau^{(1)}\right|_{\dd\Pi} = 0, \quad 
\left.\Phi_\nu^{(2)}\right|_{\dd\Pi} = 0, \quad 
\left.\left(\dd_\nu \Phi^{(1)} + (\Sigma\Phi)^{(1)}\right)_\nu\right|_{\dd\Pi} = 0,
\quad 
\left.\left(\dd_\nu \Phi^{(2)} + (\Sigma\Phi)^{(2)}\right)_\tau\right|_{\dd\Pi} = 0,
$$
where
$\Phi^{(1)} = (\Phi_1, \Phi_2, \Phi_3)$,
$\Phi^{(2)} = (\Phi_4, \Phi_5, \Phi_6)$.
Thus, the condition $ A (q, r) $ means that for any Schr\"odinger operator with potentials $ V $ and $ \Sigma $ from the considered classes there are no eigenvalues.
\end{rem}

We now state the main result.

\begin{theorem}
\label{t12}
Let $ U $ be a bounded convex domain in the plane
with piecewise $ C^2 $-smooth boundary.
Suppose that the cylinder $ \Pi = U \times \R $ satisfies the condition $ A (q, r) $
for some $ q \ge 3/2 $, $ r \ge 2 $.
Let the scalar coefficients $ \er $, $ \mu $ satisfy the condition
\eqref {01} and
$$
\er, \mu \in \widetilde W^2_p (\Pi), \qquad
p = \max \left(q, \frac{3r}{2+r}\right) .
$$
Then the spectrum of the Maxwell operator defined by the formulas
\eqref {10} -- \eqref {12} is absolutely continuous.
\end{theorem}

The Theorem \ref{t12} is conditional: {\it if the periodic Schr\"odinger operator with the Robin boundary condition in a cylinder does not have eigenvalues, then the spectrum of the Maxwell operator in the same cylinder is absolutely continuous.} The question of the absence of eigenvalues for the Schr\"odinger operator in a cylinder with an arbitrary cross-section remains open. I.~Kachkovskiy established it for cylinders with the rectangular cross-section and with the circular cross-section.

\begin{theorem}[\cite{K1}]
Let $ U $ be a rectangle on the plane. Then the condition $ A (3/2, r) $ with any $ r> 2 $ holds for the cylinder $ U \times \R $. 
\end{theorem}

\begin{theorem}[\cite{K2}]
Let $ U $ be a circle on the plane. Then the condition $ A (2, 4) $ holds for the cylinder $ U \times \R $.
\end{theorem}

Now the Theorem \ref{t12} implies

\begin{cor}
Let $ U $ be a rectangle, the coefficients $ \er $, $ \mu $ satisfy
\eqref {01} and $ \er, \mu \in \widetilde W^2_p (\Pi) $, $ p> 3/2 $.
Then the spectrum of the Maxwell operator \eqref {10} -- \eqref {12} is absolutely continuous.
\end{cor}

\begin{cor}
Let $ U $ be a disk, the coefficients $ \er $, $ \mu $ satisfy \eqref {01} and $ \er, \mu \in \widetilde W ^ 2_2 (\Pi) $.
Then the spectrum of the Maxwell operator \eqref {10} -- \eqref {12} is absolutely continuous.
\end{cor}

\begin{rem}
The Theorem \ref {t12} remains true also for the case of a nonconvex cross-section $ U $ with $ C^2 $-smooth boundary $ \dd U $, see Remark \ref {r12}.
\end{rem}

\begin{rem}
For simplicity, we formulated the theorem \ref {t12} for the {\it cylinder} $ \Pi $. It will be clear from the proof that an analogous result holds for a periodic {\it waveguide} (a domain with a variable cross-section whose change along the $ x_3 $ axis is also periodic) with $ C ^ 2 $-smooth boundary.
\end{rem}

\section{Lemmas}

Here and in what follows it is assumed that the cylinder $ \Pi $ satisfies the conditions of the Theorem \ref {t12}.

\begin{lemma}
\label{l21}
р) Let $a \in H^1 (\Pi, \C^3)$, 
$\ph \in \widetilde W^1_3 (\Pi) \cap L_\infty (\Pi)$.
Then
$$
\rot (\ph a) = \ph \rot a + [\n\ph, a] \in L_2 (\Pi, \C^3),
$$
where  $[\,.\,,\,.\,]$ is a cross product of three-dimensional vectors.

b) Let $a, d \in H^1 (\Pi, \C^3)$, 
$b \in \widetilde W^1_{3/2} (\Pi)$.
Then
$$
\<\rot [a, b], d\> = \< a \div b - b \div a 
- \< a, \n \> b + \< b, \n \> a, d\> \in L_1 (\Pi, \C^3) .
$$
\end{lemma}

Equalities are well known, and $ L_2 $- or $ L_1 $-summability follows from the embedding theorems.
\begin{lemma}
\label{l24}
Let $d \in H^1 (\Pi, \C^3)$,
$c \in H^1 (\Pi, \C^3)$ or 
$c = [a, b]$, where $a \in H^1 (\Pi, \C^3)$,\\ $b \in \widetilde W^1_{3/2} (\Pi)$.
Then
$$
\int_\Pi \< \rot c, d \> dx = 
\int_\Pi \< c, \rot d \> dx + \int_{\dd\Pi} \< c, [d, \nu] \> dS .
$$
Here $ \nu $ is the vector of unit outer normal to the boundary.
\end{lemma}

This lemma is also well known for smooth functions. The convergence of all integrals again follows from the embedding theorems.

\begin{lemma}
\label{l26}
Let $v \in H^1 (\Pi, \C^3)$, $\mu$ satisfy \eqref{01} and
$\mu \in \widetilde W^2_{3/2} (\Pi)$.
Then
\begin{eqnarray*}
\div (\mu^{-1} \n\mu) v
+ \mu^{-2} |\n\mu|^2 v 
- \<v, \n\> (\mu^{-1} \n \mu) 
= W_0 (\mu) v ,
\end{eqnarray*}
where
\begin{equation}
\label{w0}
W_0 (\mu)_{jk}= \mu^{-1} \D\mu \de_{jk} 
+ \mu^{-2} \dd_j \mu \dd_k \mu 
- \mu^{-1} \dd_j \dd_k \mu .
\end{equation}
\end{lemma}

This lemma is verified by direct computation.

On the boundary of the cylinder $ \Pi $ we introduce the function $ \kappa $ as follows: for the point
$$
x = (x_1, x_2, x_3) \in \dd U \times \R
$$ 
the value $ \kappa (x) $ is equal to the curvature of the curve $ \dd U $ at the point $ (x_1, x_2) $. In other words, $ \kappa (x) $ is the mean curvature (the sum of the principal curvatures)
of the surface $ \dd \Pi $ at the point $ x $.

\begin{lemma}
\label{l23}
Let the functions $a, b \in C^1 (\overline\Pi, \C^3)$.
On the boundary $ \dd \Pi $ we consider the function
\begin{equation}
\label{i1}
I(x) \equiv \nu_k(x) a_k(x) \dd_j \overline b_j(x) 
- \nu_j(x) a_k(x) \dd_k \overline b_j(x),
\end{equation}
$\nu(x)$ is the unit outer normal.
If the boundary conditions 
\begin{equation}
\label{i2}
\left.a_\nu\right|_{\dd\Pi} = \left.b_\nu\right|_{\dd\Pi} = 0
\end{equation}
or
\begin{equation}
\label{i3}
\left.a_\tau\right|_{\dd\Pi} = \left.b_\tau\right|_{\dd\Pi} = 0 ,
\end{equation}
are satisfied, then
$$
I(x) = \kappa (x) \< P_{e_3^\perp} a(x), b(x) \> ,
$$
where $P_{e_3^\perp}$ is the projection onto the plane orthogonal to the axis of the cylinder.
\end{lemma}

\begin{proof}
1) Consider the case \eqref{i2}.
We have
$$
I(x) = - \nu_j a_k \dd_k \overline b_j = \overline b_j a_k \dd_k \nu_j 
$$
since $ \<\nu, b \> = 0 $ on $ \dd \Pi $, and this equality can be differentiated along the tangent vector $ a $.
It is also clear that $ \dd_3 \nu_j = 0 $.
Thus,
$I(x) = \< \dd_{\tilde a} \nu, b\>$,
where $\tilde a = P_{e_3^\perp} a$.
Condition $a_\nu = 0$ implies $\tilde a = e^{i\te} |\tilde a| \tau$,
where $\te (x) \in \R$, $\tau_1 = \nu_2$, $\tau_2 = - \nu_1$.
By the FrenetЦSerret formulas
$$
\dd_{\tilde a} \nu = e^{i\te} |\tilde a| \dd_\tau \nu = 
e^{i\te} |\tilde a| \kappa (x) \tau .
$$

Therefore,
$I(x) = \kappa (x) \< \tilde a, b \>$.

2) Now consider the case \eqref{i3}.
We have
$a = e^{i\te} |a| \nu$ and
$$
\nu_j \nu_k \dd_k + \tau_j \tau_k \dd_k = \dd_j, \qquad j=1,2,
$$
where $\tau_1 = \nu_2$, $\tau_2 = - \nu_1$, $\tau_3 = \nu_3 = 0$.
Hence,
$$
I(x)= e^{i\te} |a| \left(\nu_k \nu_k \dd_j - \nu_j \nu_k \dd_k\right) \overline b_j
= e^{i\te} |a| \tau_j \tau_k \dd_k \overline b_j,
$$
where we took into account that $ b_3 = 0 $.
Furthermore, $\dd_\tau (\tau_j \overline b_j) = 0$
since $\tau_j \overline b_j = 0$ along $\dd\Pi$.
Therefore, again according to the FrenetЦSerret formulas
$$
I(x) = - e^{i\te} |a| \overline b_j \dd_\tau \tau_j  
= e^{i\te} |a| \overline b_j \kappa(x) \nu_j = \kappa (x) \< a, b \> .\quad
\qedhere
$$
\end{proof}

\begin{lemma}
\label{l22}
Let $a, b \in H^1 (\Pi, \C^3)$ and the boundary conditions 
\eqref{i2} or \eqref{i3} are satisfied.
Then
$$
\int_\Pi \< \rot a, \rot b \> dx 
= \int_\Pi \< \dd_j a, \dd_j b \> dx 
- \int_\Pi \< \div a, \div b \> dx 
+ \int_{\dd\Pi} \kappa (x) \< P_{e_3^\perp} a(x), b(x) \> dS .
$$
\end{lemma}

\begin{proof}
It suffices to prove the statement for smooth functions $ a, b $.
For such functions, it is well known that
$$
\int_\Pi \< \rot a, \rot b \> dx 
= \int_\Pi \< \dd_j a, \dd_j b \> dx 
- \int_\Pi \< \div a, \div b \> dx 
+ \int_{\dd\Pi} I (x) \, dS (x) ,
$$
 where the expression $I(x)$ is defined in \eqref{i1}.
It remains to refer to the previous lemma.
\end{proof}

\begin{lemma}
\label{l25}
Let $a, b \in H^1 (\Pi, \C^3)$, $\mu$ satisfy \eqref{01}
and $\mu \in \widetilde W_{3/2}^2 (\Pi)$.
Then
\begin{eqnarray*}
\int_\Pi \< \dd_j (\mu a), \dd_j (\mu^{-1} b) \> dx 
= \int_\Pi \< \dd_j (\mu^{1/2} a), \dd_j (\mu^{-1/2} b) \> dx 
- \int_\Pi \mu^{-1} \dd_j \mu \< \dd_j a, b \> dx \\
- \int_\Pi \left((4\mu^2)^{-1} |\n\mu|^2 + (2\mu)^{-1} \D \mu\right) \<a, b\> dx
+ \int_{\dd\Pi} (2\mu)^{-1} \dd_\nu \mu \<a, b\> dS .
\end{eqnarray*}
\end{lemma}

\begin{proof}
We differentiate the products
$\mu^{1/2} (\mu^{1/2} a)$ and $\mu^{-1/2} (\mu^{-1/2} b)$ :
\begin{eqnarray*}
\int_\Pi \< \dd_j (\mu a), \dd_j (\mu^{-1} b) \> dx 
= \int_\Pi \< \dd_j (\mu^{1/2} a), \dd_j (\mu^{-1/2} b) \> dx 
+ \int_\Pi \dd_j (\mu^{1/2}) \< a, \dd_j (\mu^{-1/2} b) \> dx \\
+ \int_\Pi \dd_j (\mu^{-1/2}) \< \dd_j (\mu^{1/2} a), b \> dx 
+  \int_\Pi \dd_j (\mu^{1/2}) \dd_j (\mu^{-1/2}) \<a, b\> dx
=: J_1 + J_2 + J_3 + J_4 .
\end{eqnarray*}
In the second term we integrate by parts:
\begin{eqnarray*}
J_2 = - \int_\Pi \< \dd_j (\dd_j (\mu^{1/2}) a), \mu^{-1/2} b \> dx 
+ \int_{\dd\Pi} \nu_j \dd_j (\mu^{1/2}) \< a, \mu^{-1/2} b \> dS \\
= - \int_\Pi \left( ((2\mu)^{-1} \D\mu - (4\mu^2)^{-1} |\n\mu|^2) \<a,b\> 
+ (2\mu)^{-1} \dd_j \mu \<\dd_j a, b\> \right) dx 
+ \int_{\dd\Pi} (2\mu)^{-1} \dd_\nu \mu \<a, b\> dS .
\end{eqnarray*}
Furthermore,
$$
J_3 = \int_\Pi \left( - \frac{|\n\mu|^2}{4\mu^2} \<a,b\> - 
\frac{\dd_j\mu}{2\mu}\<\dd_j a, b\> \right) dx, \qquad
J_4 = - \int_\Pi \frac{|\n\mu|^2}{4\mu^2} \<a,b\> dx .
$$
Adding up, we get the result.
\end{proof}

\section{Integration by parts}

\begin{theorem}
\label{t31}
Let $v, f \in H^1 (\Pi, \C^3)$, $\div (\mu v) = 0$. 
Let $\er$ and $\mu$ satisfy \eqref{01} and
$\er, \mu \in \widetilde W^2_{3/2} (\Pi)$.
Then
\begin{eqnarray}
\label{31}
\int_\Pi \er^{-1} \<\rot v, \rot f\> dx 
= \int_\Pi \<\rot (\mu v), \rot ((\er\mu)^{-1} f) \> dx 
- \int_\Pi \< \rot v, \mu [\n (\er\mu)^{-1}, f]\> dx \\
+ \int_\Pi \er^{-1} \< W_0(\mu) v, f\> dx
+ \int_\Pi (\er\mu)^{-1} \dd_j \mu \< \dd_j v, f\> dx
+ \int_{\dd\Pi} (\er\mu)^{-1} \<[\n\mu, v], [f, \nu]\> dS ,
\nonumber 
\end{eqnarray}
where $W_0 (\mu)$ is defined in \eqref{w0}.
\end{theorem}

\begin{proof}
We have
\begin{eqnarray*}
\int_\Pi \er^{-1} \<\rot v, \rot f\> dx 
= \int_\Pi \left\<\mu \rot v, \rot ((\er\mu)^{-1} f) 
- [\n(\er\mu)^{-1}, f]\right\> dx \\
= - \int_\Pi \left\<\mu\rot v,  [\n(\er\mu)^{-1}, f]\right\> dx
+ \int_\Pi \left\<\rot (\mu v), \rot ((\er\mu)^{-1} f) \right\> dx \\
- \int_\Pi \left\<[\n\mu, v], \rot ((\er\mu)^{-1} f) \right\> dx,
\end{eqnarray*}
where we used the Lemma \ref{l21} р) twice.

In the last summand, we apply the Lemma \ref{l24}:
\begin{eqnarray*}
- \int_\Pi \left\<[\n\mu, v], \rot ((\er\mu)^{-1} f) \right\> dx \\
= - \int_\Pi \left\<\rot [\n\mu, v], (\er\mu)^{-1} f \right\> dx
+ \int_{\dd\Pi} \< [\n \mu, v], [(\er\mu)^{-1} f, \nu]\> dS .
\end{eqnarray*}
In the first term on the right-hand side, we apply the Lemma \ref{l21} b) 
taking into account the equality $\div (\mu v)= 0$:
\begin{eqnarray*}
- \int_\Pi \left\<\rot [\mu^{-1} \n\mu, \mu v], (\er\mu)^{-1} f \right\> dx \\
= \int_\Pi \left\<\mu v \div (\mu^{-1} \n\mu) 
+ \<\mu^{-1} \n\mu, \n\> (\mu v) 
- \< \mu v, \n \> (\mu^{-1} \n\mu), (\er\mu)^{-1} f \right\> dx.
\end{eqnarray*}
It remains to use the Lemma \ref{l26}.
\end{proof}

Recall that our goal is to transform the Maxwell operator to the Schr\"odinger operator. The terms on the right-hand side of \eqref{31} that do not contain the derivatives $ v $ and $ f $, as well as the second summand containing $ \rot v $, are suitable for this purpose.
The first term on the right-hand side of \eqref {31} is transformed with the help of Lemmas \ref {l22} and \ref {l25} to the form
$$
\< \dd_j (\mu^{1/2} v), \dd_j (\er^{-1} \mu^{-1/2} f) \>,
$$
which corresponds to the Laplace operator.
The fourth term on the right-hand side of \eqref {31} containing $ \<\dd_j v, f \> $ cancels exactly.
It explains the choice of the exponents of $ \er $ and $ \mu $ in the expression
$\< \dd_j (\mu^{1/2} v), \dd_j (\er^{-1} \mu^{-1/2} f) \>$.

\begin{lemma}
\label{c32}
Let $v, f \in H^1 (\Pi, \C^3)$,
$\div (\mu v) = 0$,
$\left.v_\nu\right|_{\dd\Pi} = 0$,
$\left.f_\nu\right|_{\dd\Pi} = 0$.
Let $\er$ and $\mu$ satisfy \eqref{01}
and $\er, \mu \in \widetilde W_{3/2}^2 (\Pi)$.
Then
\begin{eqnarray*}
\int_\Pi \er^{-1} \<\rot v, \rot f\> dx 
= \int_\Pi \< \dd_j (\mu^{1/2} v), \dd_j (\er^{-1} \mu^{-1/2} f) \> dx 
- \int_\Pi \< \rot v, \mu [\n (\er\mu)^{-1}, f]\> dx \\
+ \int_\Pi \er^{-1} \< W(\mu) v, f\> dx
+ \int_{\dd\Pi} \left\<\left(\er^{-1} \kappa (x) P_{e_3^\perp}
- (2\er\mu)^{-1} \dd_\nu \mu\right) v, f\right\> dS .
\end{eqnarray*}
Here
\begin{equation}
\label{32}
W_{jk} (\mu) = \left((2\mu)^{-1} \D\mu - (4\mu^2)^{-1} |\n\mu|^2\right) \de_{jk} 
+ \mu^{-2} \dd_j \mu \dd_k \mu - \mu^{-1} \dd_j \dd_k \mu.
\end{equation}
\end{lemma}

\begin{proof}
By the Lemma \ref {l22} taking into account $ \div (\mu v) = 0 $ we have
$$
\int_\Pi \<\rot (\mu v), \rot ((\er\mu)^{-1} f)\> dx 
= \int_\Pi \<\dd_j (\mu v), \dd_j ((\er\mu)^{-1} f)\> dx
+ \int_{\dd\Pi} \kappa (x) \er^{-1} \< P_{e_3^\perp} v, f \> dS .
$$
By the Lemma \ref{l25} with $a=v$, $b=\er^{-1} f$ we have
\begin{eqnarray}
\label{**}
\int_\Pi \<\dd_j (\mu v), \dd_j ((\er\mu)^{-1} f)\> dx
= \int_\Pi \<\dd_j (\mu^{1/2} v), \dd_j (\er^{-1} \mu^{-1/2} f)\> dx 
- \int_\Pi (\er\mu)^{-1} \dd_j \mu \< \dd_j v, f\> dx \\
- \int_\Pi \left((4\er\mu^2)^{-1} |\n\mu|^2 + (2\er\mu)^{-1} \D \mu\right) \<v, f\> dx
+ \int_{\dd\Pi} (2\er\mu)^{-1} \dd_\nu \mu \<v, f\> dS .
\nonumber
\end{eqnarray}
Finally, it follows from the condition $\left.v_\nu\right|_{\dd\Pi} = 0$ that
$$
\<[\n\mu, v], [f, \nu]\> 
= \<[\nu, [\n\mu, v]], f\> = 
- \dd_\nu \mu \<v, f\>.
$$ 
Therefore, the last integral on the right-hand side of \eqref {31} is equal to
$$
\int_{\dd\Pi} (\er\mu)^{-1} \<[\n\mu, v], [f, \nu]\> dS
= - \int_{\dd\Pi} (\er\mu)^{-1} \dd_\nu \mu \<v, f\> dS . 
\quad \qedhere
$$
\end{proof}

\begin{lemma}
\label{c33}
Let $u, w \in H^1 (\Pi, \C^3)$,
$\div (\er u) = 0$,
$\left.u_\tau\right|_{\dd\Pi} = 0$,
$\left.w_\tau\right|_{\dd\Pi} = 0$.
Let $\er$ and $\mu$ satisfy \eqref{01}
and $\er, \mu \in \widetilde W_{3/2}^2 (\Pi)$.
Then
\begin{eqnarray*}
\int_\Pi \mu^{-1} \<\rot u, \rot w\> dx 
= \int_\Pi \< \dd_j (\er^{1/2} u), \dd_j (\er^{-1/2} \mu^{-1} w) \> dx 
- \int_\Pi \< \rot u, \er [\n (\er\mu)^{-1}, w]\> dx \\
+ \int_\Pi \mu^{-1} \< W(\er) u, w\> dx
+ \int_{\dd\Pi} \left(\mu^{-1} \kappa (x) 
+ (2\er\mu)^{-1} \dd_\nu \er \right) \<u, w\> dS ,
\end{eqnarray*}
where the matrix $ W (\er) $ is given by the formula \eqref {32}.
\end{lemma}

\begin{proof}
We apply the Theorem \ref {t31} with $ v = u $, $ f = w $ and with $ \er $ and $ \mu $ interchanged.
We get
\begin{eqnarray*}
\int_\Pi \mu^{-1} \<\rot u, \rot w\> dx 
= \int_\Pi \< \rot (\er u), \rot ((\er\mu)^{-1} w) \> dx 
- \int_\Pi \< \rot u, \er [\n (\er\mu)^{-1}, w]\> dx \\
+ \int_\Pi \mu^{-1} \< W_0 (\er) u, w\> dx
+ \int_{\Pi} (\er\mu)^{-1} \dd_j \er \<\dd_j u, w\> dx .
\end{eqnarray*}
The last integral on the right-hand side of \eqref {31} is zero
since  $\left.w_\tau\right|_{\dd\Pi} = 0$.

We again apply the Lemma \ref {l22} to the first integral on the right-hand side, taking into account $ \div (\er u) = 0 $:
$$
\int_\Pi \<\rot (\er u), \rot ((\er\mu)^{-1} w)\> dx 
= \int_\Pi \<\dd_j (\er u), \dd_j ((\er\mu)^{-1} w)\> dx
+ \int_{\dd\Pi} \kappa (x) \mu^{-1} \< u, w \> dS .
$$
By the Lemma \ref {l25} with $ a = u $, $ b = \mu^{-1} w $ we obtain similarly to \eqref {**}
\begin{eqnarray*}
\int_\Pi \<\dd_j (\er u), \dd_j ((\er\mu)^{-1} w)\> dx
= \int_\Pi \<\dd_j (\er^{1/2} u), \dd_j (\er^{-1/2} \mu^{-1} w)\> dx 
- \int_\Pi (\er\mu)^{-1} \dd_j \er \< \dd_j u, w\> dx \\
- \int_\Pi \left((4\er^2\mu)^{-1} |\n\er|^2 + (2\er\mu)^{-1} \D \er\right) \<u, w\> dx
+ \int_{\dd\Pi} (2\er\mu)^{-1} \dd_\nu \er \<u, w\> dS ,
\end{eqnarray*}
and now the statement follows.
\end{proof}

\section{Proof of the Theorem \ref{t12}}
It suffices to prove the absence of eigenvalues.
Suppose that $ \left\{ u, v \right\} $ is an eigenfunction of the Maxwell operator, that is,
\begin{equation}
\label{41}
- i \mu^{-1} \rot u = \la v, \quad \div (\er u) = 0, \quad 
\left.u_\tau\right|_{\dd\Pi} = 0,
\end{equation}
\begin{equation}
\label{42}
i \er^{-1} \rot v = \la u, \quad \div (\mu v) = 0, \quad 
\left.v_\nu\right|_{\dd\Pi} = 0.
\end{equation}
Let $w \in H^1 (\Pi, \C^3)$, $\left.w_\tau\right|_{\dd\Pi} = 0$.
Then from \eqref {41}, \eqref {42} and the Lemma \ref {l24} we obtain 
\begin{equation}
\label{43} 
\int_\Pi \mu^{-1} \< \rot u, \rot w \> dx 
= i \la \int_\Pi \< v, \rot w \> dx 
= i \la \int_\Pi \< \rot v, w \> dx
= \la^2 \int _\Pi \er \< u, w \> dx .
\end{equation}
Similarly, if $f \in H^1 (\Pi, \C^3)$, $\left.f_\nu\right|_{\dd\Pi} = 0$, then
\begin{equation}
\label{44} 
\int_\Pi \er^{-1} \< \rot v, \rot f \> dx 
= - i \la \int_\Pi \< u, \rot f \> dx 
= - i \la \int_\Pi \< \rot u, f \> dx
= \la^2 \int _\Pi \mu \< v, f \> dx .
\end{equation}
In the middle equalities in \eqref{43} and \eqref{44}
we used the Lemma \ref{l24} and the conditions $\left.w_\tau\right|_{\dd\Pi} = 0$ and $\left.u_\tau\right|_{\dd\Pi} = 0$.
We add the equalities \eqref {43} and \eqref {44}, substituting the expressions from the Lemmas \ref {c32} and \ref {c33} instead of the left-hand sides:
\begin{eqnarray*}
\int_\Pi \left(\< \dd_j (\er^{1/2} u), \dd_j (\er^{-1/2} \mu^{-1} w) \>
+ \< \dd_j (\mu^{1/2} v), \dd_j (\er^{-1} \mu^{-1/2} f) \>\right.\\
\left.- \< \rot u, \er [\n (\er\mu)^{-1}, w]\> 
- \< \rot v, \mu [\n (\er\mu)^{-1}, f]\>
+ \mu^{-1} \< W(\er) u, w\> + \er^{-1} \< W(\mu) v, f\>\right) dx \\
+ \int_{\dd\Pi} \left(\left(\mu^{-1} \kappa (x) 
+ (2\er\mu)^{-1} \dd_\nu \er \right) \<u, w\> 
+ \left\<\left(\er^{-1} \kappa (x) P_{e_3^\perp}
- (2\er\mu)^{-1} \dd_\nu \mu\right) v, f \right\>\right) dS \\
= \la^2 \int_\Pi \left(\er \<u, w\> + \mu \<v, f\>\right) dx .
\end{eqnarray*}

Taking into account that $ \rot u $ and $ \rot v $ in the first integral on the left-hand side can be expressed from \eqref {41} and \eqref {42}, the last equality can be rewritten as
\begin{equation}
\label{45} 
\int_\Pi \left(\< \dd_j \Phi, \dd_j \Psi \>_{\C^6} 
+ \<V \Phi, \Psi\>_{\C^6}\right) dx 
+ \int_{\dd\Pi} \<\Sigma \Phi, \Psi\>_{\C^6} dS = 0,
\end{equation}
where
$$
\Phi = 
\left( \begin{array}{cc} \er^{1/2} u \\ \mu^{1/2} v \end{array} \right) , \quad 
\Psi = \left( \begin{array}{cc} \er^{-1/2} \mu^{-1} w \\ 
\er^{-1} \mu^{-1/2} f \end{array} \right), \qquad 
V = \left( \begin{array}{cc}
W(\er) - \er \mu \la^2 I_3 & -2 i \la F \\
2 i \la F & W(\mu) - \er \mu \la^2 I_3
\end{array} \right) ,
$$
$$
F = \left( \begin{array}{ccc}
0 & - \dd_3 (\er \mu)^{1/2} & \dd_2 (\er \mu)^{1/2} \\
\dd_3 (\er \mu)^{1/2} & 0 & - \dd_1 (\er \mu)^{1/2} \\
- \dd_2 (\er \mu)^{1/2} & \dd_1 (\er \mu)^{1/2} & 0
\end{array} \right) , 
$$
$$ 
\Sigma = \left( \begin{array}{cc}
\left(\kappa (x) + (2\er)^{-1} \dd_\nu \er\right) I_3 & 0 \\
0 & \kappa (x) P_{e_3^\perp} - (2\mu)^{-1} \dd_\nu \mu I_3
\end{array} \right) .
$$
The conditions $ \er, \mu \in \widetilde W_q^2 $ for $ q \ge 3/2 $ and the condition \eqref {01} imply $ W (\er), W (\mu) \in \widetilde { L_q} $ and $ V \in \widetilde {L_q} (\Pi) $.

It follows from $\er, \mu \in \widetilde W_{\frac{3r}{2+r}}^2$ that
$\dd_\nu \er, \dd_\nu \mu, \Sigma \in \widetilde{L_r} (\dd\Pi)$. 
The identity \eqref {45} holds for any $ \Psi \in \hat H^1 (\Pi) $, therefore, according to the condition $ A (q, r) $, we obtain $ \Phi \equiv 0 $. It means that $ u \equiv v \equiv 0 $
and the point spectrum of the Maxwell operator is empty.
The Theorem \ref {t12} is proved.


\end{document}